\documentclass[10pt,conference,compsocconf,onecolumn]{IEEEtran}

\usepackage[cmex10]{amsmath}
\usepackage{amssymb}
\usepackage{amsthm}
\usepackage{booktabs}
\usepackage{cite}
\ifCLASSINFOpdf
  \usepackage[pdftex]{graphicx}
  \DeclareGraphicsExtensions{.pdf,.jpeg,.png}
\else
  \usepackage[dvips]{graphicx}
  \DeclareGraphicsExtensions{.eps}
\fi
\usepackage[caption=false,font=footnotesize]{subfig}
\usepackage{url}

\begin{document}

\newcommand{\calg}[1]{{\mathcal #1}}
\newcommand{\bfg}[1]{\text{\boldmath $#1$}}
\newcommand{\dcsAlgRef}[1]{Algorithm~\ref{alg:#1}} 
\newcommand{\dcsSecRef}[1]{Section~\ref{#1}} 
\newcommand{\dcsSSecRef}[1]{Section~\ref{ssec:#1}} 
\newcommand{\dcsSSSecRef}[1]{Section~\ref{sssec:#1}} 
\newcommand{\dcsFigRef}[1]{Figure~\ref{fig:#1}} 
\newcommand{\dcsTabRef}[1]{Table~\ref{#1}} 
\newcommand{\dcsEqRef}[1]{Eq.~\eqref{eq:#1}} 
\newcommand{\dcsIneqRef}[1]{Eq.~\eqref{eq:#1}} 
\newcommand{\dcsHRule}{\rule{\linewidth}{0.08em}}
\newtheorem{proposition}{Proposition}
\newtheorem{example}{Example}
\newtheorem{scenario}{Scenario}
\newtheorem{thm}{Theorem}[section]

\begin{center}
{\Huge A Game-Theoretic Approach to Distributed Coalition Formation in Energy-Aware Cloud Federations (Extended Version)}
\end{center}

\begin{center}
{\Large\bfseries
\ \\ 
Please, cite this paper as:
\begin{tabular}{c}
Marco Guazzone, Cosimo Anglano, Matteo Sereno,\\
\emph{``A Game-Theoretic Approach to Coalition Formation},\\
\emph{in Green Cloud Federations,''}\\
Proc. of the $14^\text{th}$ IEEE/ACM International Symposium\\
on Cluster, Cloud and Grid Computing (CCGrid), pp. 618--625, 2014.\\
DOI: \url{10.1109/CCGrid.2014.37}\\
{\normalsize Publisher: \url{http://doi.ieeecomputersociety.org/10.1109/CCGrid.2014.37}}
\end{tabular}
}
\end{center}

\newpage

\title{A Game-Theoretic Approach to Distributed Coalition Formation in Energy-Aware Cloud Federations (Extended Version)}

\author{
\IEEEauthorblockN{Marco Guazzone\IEEEauthorrefmark{1}, Cosimo Anglano\IEEEauthorrefmark{2}, Matteo Sereno\IEEEauthorrefmark{1}}
\IEEEauthorblockA{\IEEEauthorrefmark{1}Department of Computer Science, University of Torino, Italy}
\IEEEauthorblockA{\IEEEauthorrefmark{2}Department of Science and Technological Innovation, University of Piemonte Orientale, Italy}
}
\IEEEspecialpapernotice{Please, cite as:\\ \textbf{Marco Guazzone, Cosimo Anglano, Matteo Sereno,\\ ``Game-Theoretic Approach to Coalition Formation in Green Cloud Federations,''\\ Proc. of the $14^\text{th}$ IEEE/ACM International Symposium on Cluster, Cloud and Grid Computing (CCGrid), pp. 618--625, 2014,\\ DOI: \url{10.1109/CCGrid.2014.37},\\ Publisher: \url{http://doi.ieeecomputersociety.org/10.1109/CCGrid.2014.37}.}}
\maketitle

\begin{abstract}
Federations among sets of Cloud Providers (CPs), whereby a set of CPs agree to
mutually use their own resources to run the VMs of other CPs, are considered a
promising solution to the problem of reducing the energy cost.
In this paper, we address the problem of federation formation for a set of CPs,
whose solution is necessary to exploit the potential of cloud federations for
the reduction of the energy bill.
We devise a distributed algorithm, based on cooperative game theory, that allows
a set of CPs to cooperatively set up their federations in such a way that their
individual profit is increased with respect to the case in which they work in
isolation, and we show that, by using our algorithm and the proposed CPs'
utility function, they are able to self-organize into Nash-stable federations
and, by means of iterated executions, to adapt themselves to environmental
changes.
Numerical results are presented to demonstrate the effectiveness of the
proposed algorithm.
\end{abstract}


\IEEEpeerreviewmaketitle

\section{Introduction}

Many modern Internet services are implemented as cloud
applications consisting of a set of \emph{Virtual Machines} (VMs)
that are allocated and run on a physical computing infrastructure
managed by a virtualization platform (e.g., Xen~\cite{xen}, VMware~\cite{vmware}, etc.).
These infrastructure are typically owned
by a \emph{Cloud Provider} (CP) (e.g., Amazon AWS, Rackspace, Windows Azure, etc.),
and are located into a (set of possibly distributed) data center(s).

One of the key issues that must be faced by a CP is the
reduction of its energy cost, that represents
a large fraction of the total cost of ownership for
physical computing infrastructures~\cite{energystar}.
This cost is mainly due to the consumption of the
physical resources that must be switched on to run the
workload.

To reduce energy consumption, two techniques are therefore possible for a CP:
(a) to minimize the number of hosts that are switched on by maximizing
the number of VMs allocated on each physical resource (using
suitable resource management techniques~\cite{Guazzone-2011-EEResMgmt,Albano-2013}),
and (b) to use resources that consume less energy.

Cloud federations~\cite{rochwerger:reservoir-computer}, whereby a set of CPs agree to mutually
use their own resources to run the VMs of other CPs, are considered to be a promising solution
for the reduction of  energy costs~\cite{Celesti-2013} as they ease the application
of both techniques.

As a matter of fact, while each individual CP is bound to its specific energy provider
and to the physical infrastructure it owns, a set of federated CPs may enable
the usage of more flexible energy management strategies that,
by suitably relocating the workload towards CPs that pay less for the energy, or that have more 
energy-efficient resources, may reduce the energy bill for each one of them. 

In order to exploit the energy saving potential of cloud federations,
it is however necessary to address the question concerning its formation.
As a matter of fact, it is unreasonable to assume that a CP unconditionally
joins a federation regardless of the benefits it receives, while it
is reasonable to expect that it joins a federation only if this brings it
a benefit.

In this paper, we address the problem of federation formation for a set of CPs,
and we devise an algorithm that allows these CPs to decide whether to federate
or not on the basis of the profit they receive for doing so.
In our approach, each CP pays for the energy consumed by each VM, whether it
belongs to its own workload or to the one of another CP, but receives a payoff
(computed as discussed later) for doing so.

The algorithm we propose is based on cooperative game
theory~\cite{Peleg-2007-CooperativeGames,Book_RAY2007}. 
In particular, we rely on \emph{hedonic games} (see \cite{BOGOMONLAIA-JACKSON_2002} for 
their definition)
whereby each CP bases its decision on its own preferences.
Depending on the specific operational conditions of each CP
(i.e., the resource requirements of its workload, its cost of energy,
the energy consumption of its physical machines, and the revenue it obtains
when running each VM on these machines), different
federations (each one consisting of a subset of the CPs), or even no federation at all, may be
formed by the involved CPs.
We call \emph{federation set} the set of distinct federations
formed by a set of CPs.

The algorithm we propose computes the federation set that
results in the highest profit that can be achieved by a set of autonomous and selfish CPs.
This derives from the fact that this algorithm ensures that all the federations
formed by groups of CPs are \emph{stable}, that is CPs have no incentive
to leave the federation once they decide to participate.

Unlike similar proposals (e.g., \cite{Mashayekhy-2012-Coalitional}), that
rely on a centralized architecture in which a trusted third party computes
the federation set,
we adopt a distributed approach in which each CP
autonomously and selfishly makes its own decisions,
and the best solution emerges from these decisions without the
need of synchronizing them, or to resort to a trusted third party.
In this way, we avoid two drawbacks that affect existing proposals,
namely the difficulty of finding a third party that is trusted by
all the CPs, and the need to suspend the operations of all the CPs
when the federation set is being computed.

The rest of this paper is organized as follows.
In \dcsSecRef{FORMULATION}, we describe the system under study, and provide some
simple motivating example.
In \dcsSecRef{GAMETHEORY1}, we present a cooperative game-theoretic model of the
system under study, and show stability conditions and profit allocation strategies
that provide the theoretical foundation for the distributed coalition formation
algorithm, that is also presented in this section. 
In \dcsSecRef{NUMERICAL}, we present results from an experimental evaluation to show the effectiveness of the proposed approach.
In \dcsSecRef{sec:related}, we provide a short overview of related works.
Finally, conclusions and an outlook on future extensions are presented in \dcsSecRef{CONCL}.

\section{Problem Formulation\label{FORMULATION}}

In this section, we first formally describe the problem addressed in the
paper (see \dcsSecRef{SYSTEM}), and then we illustrate some issues that must
be properly addressed in order to properly solve it (see \dcsSecRef{EXAM1}).
\subsection{System Description\label{SYSTEM}}

We consider a set of $n$ CPs denoted by
$\calg{N}=\bigl\{1,2,\ldots,n\bigr\}$, where each CP $i$ is endowed with a set
$\calg{H}_i$ of physical hosts.
We denote as $\calg{H}=\calg{H}_1\cup\calg{H}_2\cup\cdots\cup\calg{H}_n$
the set of all the hosts collectively belonging to the various CPs.
These hosts are grouped into a set $\calg{G}$
of \emph{host classes} according to their processor type and to the amount of
physical memory they provide; all the hosts in the same class are homogeneous
in terms of processor and memory size.
For any $h \in \calg{H}$ we denote by $g(h)$ the function that gives the
host class of $h$ (i.e., a function $g: \calg{H} \rightarrow  \calg{G}$).

As discussed in~\cite{Rivoire-2008-Comparison}, we assume that a host $h$ consumes
$C^{\text{min}}_{g(h)}$ Watts when its CPU is in the idle state, $C^{\text{max}}_{g(h)}$ Watts
when its CPU is fully utilized, and
$\Bigl(C^{\text{min}}_{g(h)}+f\cdot\bigl(C^{\text{max}}_{g(h)}-C^{\text{min}}_{g(h)}\bigr)\Bigr)$ when
a fraction $f\in[0,1]$ of its CPU capacity is used.
This model, albeit simple, has been shown to provide accurate estimates of
power consumption for different host types when running several benchmarks
representative of real-world applications~\cite{Rivoire-2008-Comparison}.

Physical hosts run cloud workloads, consisting in a set
$\calg{J}=\calg{J}_1 \cup \calg{J}_2 \cup \cdots \cup \calg{J}_n$ of
VMs, where $\calg{J}_i$ denotes the set of VMs that
compose the workload of the $i$-th CP (each VM contains the whole execution
environment of one or more applications).

As typically done by CPs,
VMs are grouped into a set $\calg{Q}$ of \emph{VM classes} according to
the computing capacity provided by their virtual processors, and to the amount of
physical memory they are equipped with; all the VMs belonging to the same class
provide the same amount of computing capacity and of physical memory.
For instance, Amazon's EC2~\cite{EC2} defines the \emph{Elastic Compute
Unit (ECU)} as an abstract computing resource able to deliver a capacity
equivalent to that of a 1.2 GHz 2007 Xeon processor, and provides
various \emph{instance types} (that are equivalent to our VM classes) that differ
among them in the number of ECUs and in the amount of RAM they are equipped with.
More specifically, \emph{small}, \emph{medium}, and \emph{large}, corresponding
to VM class 1, 2, and 3, respectively, provide 1 ECU and 1.7 GB of RAM,
2 ECUs and 3.7 GB of RAM, and 4 ECUs and 7.5 GB of RAM, respectively.

For any VM $j \in \calg{J}$, we denote by $q(j)$ the function that gives its
VM class (i.e., a function $q: \calg{J} \rightarrow  \calg{Q}$).
Using this notation, for any $j \in \calg{J}$ we denote by
$\mathit{CPU}_{q(j)}$ and by $\mathit{RAM}_{q(j)}$ the amount of
computing capacity and of physical memory of VM $j$, respectively.
As an example, in the Amazon EC2 case we have that
$\mathit{CPU}_{1} = 1$ ECU and $\mathit{RAM}_{1}= 1.7$ GB, while
$\mathit{CPU}_{3} = 4$ ECU and $\mathit{RAM}_{3} = 7.5$ GB.

When allocated on a physical host $h$, a VM $j$ uses a certain fraction
$A_{q(j),g(h)}$ of CPU capacity and a certain fraction $M_{q(j),g(h)}$ of
physical memory.
$A_{q(j),g(h)}$ can be determined by measuring,
with a suitable benchmark (e.g., GeekBench~\cite{GeekBench}), the computing capacity $\mathit{Cap}_{v}$ delivered by
the virtual processor of VMs in $q(j)$
and the capacity $\mathit{Cap}_{p}$ delivered by the physical processor of hosts in $g(h)$, and then by dividing these quantities,
i.e., $A_{q(j),g(h)} = \frac{\mathit{Cap}_{v}}{\mathit{Cap}_{p}}$.
For instance, if $\mathit{Cap}_{v} = 1,000$ and $\mathit{Cap}_{p}=8,000$, then
$A_{q(j),g(h)} = 0.125$.
$M_{q(j),g(h)}$ can instead be computed as $M_{q(j),g(h)} = \Bigl\lceil\frac{\mathit{RAM}_{q(j)}}{\text{RAM size of hosts in $g(h)$}}\Bigr\rceil$.

Each CP $i$ charges, for each VM $j$, a \emph{revenue rate} (that depends on the
class $q(j)$ of that VM) that specifies the amount of money that the VM owner
must correspond per unit of time.
For instance, Amazon charges $0.08$ \$/hour, $0.16$ \$/hour, and $0.32$ \$/hour
for \emph{small}, \emph{medium}, and \emph{large} instances, respectively.
Consequently, CP $i$ earns a global revenue rate that is given by the sum of the
revenue rates of individual VMs.
To run its workload, CP $i$ incurs an energy cost quantified by the \emph{energy
cost rate} (the amount of money that is paid per unit of time), which is the
energy cost resulting from the allocation of (a subset of) the workload $\calg{J}$ on
its host set $\calg{H}_i$ that must be paid per unit of time (see
\dcsSecRef{OPTIMIZATION} for a discussion on the optimization technique we use to minimize it).
We define the \emph{net profit rate} of CP $i$ as the difference
between its global revenue rate (obtained for hosting a set of VMs) and
its global energy cost rate (that it incurs to run such VMs).

Our goal is to allocate all the VMs in $\calg{J}$ on the hosts in $\calg{H}$
(independently from the corresponding CPs) in such a way to maximize
the \emph{net profit rate} of each CP $i$.
This goal can be achieved by finding the smallest set of hosts that are sufficient
to accommodate the resource shares of all the VMs in $\calg{J}$ such that the overall
energy consumption is minimized, and by providing a suitable revenue for those CPs
that host VMs belonging to other CPs.

\subsection{Issues in coalition formation}\label{EXAM1}

The most straightforward way to form a coalition\footnote{In this paper, the terms
\emph{federation} and \emph{coalition} (which is widely adopted in the game-theoretic community) are used interchangeably.} is to include in it
\emph{all} the CPs (the \emph{grand coalition}).
This solution is certainly attractive because of its simplicity and ease of
implementation, and can bring significant benefits to its participant.

To illustrate, let us consider three different CPs,
named $\text{CP}_1$, $\text{CP}_2$, and $\text{CP}_3$, whose operational scenarios
are characterized by the values shown in \dcsTabRef{tbl:exam1},
where we report the characteristics of the host classes (\dcsTabRef{tbl:exam1}(a)),
of the VM classes (\dcsTabRef{tbl:exam1}(b)), and the resource shares
of each VM class (\dcsTabRef{tbl:exam1}(c)).
\begin{table}
\centering
\caption{Operational scenarios of $\text{CP}_1$, $\text{CP}_2$, $\text{CP}_3$}\label{tbl:exam1}
\begin{tabular}{rcrc} 
\toprule
\multicolumn{4}{c}{\emph{(a) Characteristics of host classes}} \\
\multicolumn{1}{c}{Host Class} & \multicolumn{1}{c}{CPU} & \multicolumn{1}{c}{RAM}  & \multicolumn{1}{c}{$C^{\text{min}}\text{/}C^{\text{max}}$} \\
                          &                         & \multicolumn{1}{c}{(GB)} & \multicolumn{1}{c}{(W)} \\
\midrule
$1$ & $2\times$ Xeon 5130 & $16$ & $ 86.7\text{/}274.9$ \\
$2$	& Xeon X3220          & $32$ & $143.0\text{/}518.4$ \\
$3$	& $2\times$ Xeon 5160 & $64$ & $490.1\text{/}1,117.8$ \\
\bottomrule
\multicolumn{4}{c}{\emph{(b) Characteristics of VM classes}} \\
\multicolumn{1}{c}{VM Class} & \multicolumn{1}{c}{Processor} & \multicolumn{1}{c}{\#CPUs} & \multicolumn{1}{c}{RAM} \\
                          &                               &                            & \multicolumn{1}{c}{(GB)} \\
\midrule
$1$ & AMD Opteron 144 & $1$ & $1$ \\
$2$ & AMD Opteron 144 & $2$ & $2$ \\
$3$ & AMD Opteron 144 & $4$ & $4$ \\
\bottomrule
\multicolumn{4}{c}{\emph{(c)Per-VM physical resource shares}} \\ 
\multicolumn{1}{c}{Host Class} & \multicolumn{1}{c}{Class-$1$ VM} & \multicolumn{1}{c}{Class-$2$ VM} & \multicolumn{1}{c}{Class-$3$ VM} \\
\midrule
$1$ & $(0.20,0.062500)$ & $(0.4,0.12500)$ & $(0.8,0.2500)$ \\
$2$ & $(0.15,0.031250)$ & $(0.3,0.06250)$ & $(0.6,0.1250)$ \\
$3$ & $(0.10,0.015625)$ & $(0.2,0.03125)$ & $(0.3,0.0625)$ \\
\bottomrule
\end{tabular}
\end{table}
Assume, for the moment, that we have both a way to compute the set of hosts
that must be switched on to minimize energy consumption (a suitable
optimization problem is presented in \dcsSecRef{OPTIMIZATION}), and
a profit distribution rule that yields suitable revenues to CPs hosting
external VMs, so that the minimization of the energy consumption within
a federation of CPs corresponds to the maximization of their net profit rates
(such a rule is discussed in \dcsSecRef{sec:CP-game}).

Now, let us consider a simple scenario (that we name
\emph{Scenario $1$}) in which each CP owns $30$ hosts of a single class,
and in particular that $\text{CP}_1$, $\text{CP}_2$, and $\text{CP}_3$ own only class-$1$, class-$2$,
and class-$3$ hosts, respectively;
furthermore, all the CPs have the same workload
(in particular, each CP has to allocate $10$ class-$3$ VMs) and
energy cost ($0.4$ \$/kWh).

If each CP uses only its own resources and allocate its own
workload (i.e., no federation is formed), it achieves the energy cost
rate reported in \dcsTabRef{tbl:ex1}(a), where also the total
cost rate is reported.
If, conversely, they form a grand coalition (i.e, they jointly
perform a global workload allocation using the union of their respective host
sets), their corresponding individual and overall energy cost rates are reported in \dcsTabRef{tbl:ex1}(b).
\begin{table}
\centering
\caption{Scenario $1$ results}\label{tbl:ex1}
\begin{tabular}{lrrr}
\toprule
\multicolumn{1}{c}{CP} & \multicolumn{1}{c}{Powered-on} & \multicolumn{1}{c}{Consumed Power} & \multicolumn{1}{c}{Energy Cost} \\ 
	                    & \multicolumn{1}{c}{Hosts}      & \multicolumn{1}{c}{(kW)} & \multicolumn{1}{c}{(\$/hour)} \\ 
	                    
\midrule
\multicolumn{4}{c}{\emph{(a) no federation among CPs}} \\ 
\midrule
$\text{CP}_1$ & $10$ & $2.37$ & $0.95$ \\ 
$\text{CP}_2$ & $10$ & $3.68$ & $1.47$ \\ 
$\text{CP}_3$ & $ 4$ & $3.84$ & $1.54$ \\ 
\midrule
\emph{Total}   & $24$ & $9.89$ & $3.96$ \\ 
\bottomrule
\multicolumn{4}{c}{\emph{(b) federation among all the CPs}} \\ 
\midrule
$\text{CP}_1$ & $30$ & $7.12$ & $2.85$ \\ 
$\text{CP}_2$ & $ 0$ & $0.00$ & $0.00$ \\ 
$\text{CP}_3$ & $ 0$ & $0.00$ & $0.00$ \\ 
\midrule
\emph{Total}   & $30$ & $7.12$ & $2.85$ \\ 
\bottomrule
\end{tabular}
\end{table}

As can be seen from these results, the grand coalition yields a smaller
total cost rate of energy that, as discussed before in \dcsSecRef{SYSTEM}, corresponds to a larger
net profit rates for the individual CPs.
In particular, the overall
energy cost rate is reduced by $28\%$ (from $3.96$ \$/hour to $2.85$ \$/hour),
thanks to the fact that in the federation case only the hosts belonging
to $\text{CP}_1$ are used to run all the VMs.

Given the benefits resulting from the grand coalition in this example,
it is natural to speculate whether the problem of computing
a federation set really arises in practice. 

Unfortunately, as shown below,
there are cases when the grand coalition does not represent the best solution
for all the involved CPs.
Indeed, consider another scenario (that we name \emph{Scenario $2$})
where the characteristics of the hosts and of the VMs are the same as before (see
\dcsTabRef{tbl:ex2}), but
the number and type of hosts and VMs differ from those assumed in \emph{Scenario $1$}.
In particular, now $\text{CP}_1$ owns $42$ class-$2$ hosts and its workload consists in
$65$ class-$2$ VMs, while $\text{CP}_2$ and $\text{CP}_3$ own $41$ class-$3$ hosts each and both
have $61$ class-$2$ VMs as workload.

The individual and overall energy cost rates corresponding to the optimal
solution for the no federation and the grand
coalition cases are reported in \dcsTabRef{tbl:ex2}(a) and (b), respectively.
\begin{table}
\centering
\caption{Scenario $2$ results}\label{tbl:ex2}
\begin{tabular}{lrrr}
\toprule
\multicolumn{1}{c}{CP} & \multicolumn{1}{c}{Powered-on} & \multicolumn{1}{c}{Consumed Power} & \multicolumn{1}{c}{Energy Cost} \\ 
                        & \multicolumn{1}{c}{Hosts}      & \multicolumn{1}{c}{(kW)} & \multicolumn{1}{c}{(\$/hour)} \\ 
\midrule
\multicolumn{4}{c}{\emph{(a) no federation among CPs}} \\
\midrule
$\text{CP}_1$ & $22$ & $10.47$ & $ 4.19$ \\
$\text{CP}_2$ & $13$ & $14.03$ & $ 5.61$ \\ 
$\text{CP}_3$ & $13$ & $14.03$ & $ 5.61$ \\ 
\midrule
\emph{Total}   & $48$ & $38.53$ & $15.41$ \\ 
\bottomrule
\multicolumn{4}{c}{\emph{(b) federation among all CPs}} \\
\midrule
$\text{CP}_1$ & $41$ & $19.71$ & $ 7.89$ \\ 
$\text{CP}_2$ & $ 9$ & $ 9.93$ & $ 3.97$ \\ 
$\text{CP}_3$ & $ 4$ & $ 4.47$ & $ 1.79$ \\
\midrule
\emph{Total}   & $54$ & $34.11$ & $13.65$ \\ 
\bottomrule
\multicolumn{4}{c}{\emph{(c) $\text{CP}_1$ and $\text{CP}_2$ federated among them, $\text{CP}_3$ alone}} \\
\midrule
$\text{CP}_1$ & $42$ & $20.20$ & $ 8.08$ \\ 
$\text{CP}_2$ & $ 0$ & $ 0.00$ & $ 0.00$ \\ 
$\text{CP}_3$ & $13$ & $14.03$ & $ 5.61$ \\ 
\midrule
\emph{Total}   & $84$ & $34.23$ & $13.69$ \\ 
\bottomrule
\end{tabular}
\end{table}
Again, we observe a reduction of the energy cost in the grand coalition case
(see\dcsTabRef{tbl:ex2}(a) and (b)), although this reduction is smaller
than in the \emph{Scenario $1$} case (it amounts to $11.4\%$).

However, by looking at the results in \dcsTabRef{tbl:ex2}(c), we can
observe that if $\text{CP}_1$ and $\text{CP}_2$ federate among them and exclude
$\text{CP}_3$, their overall energy consumption rate \emph{is smaller} than
in the case of the grand coalition, although the overall cost involving
all three CPs is higher.
As a matter of fact, in the grand coalition the joint cost rate of
$\text{CP}_1$ and $\text{CP}_2$ amounts to  $7.89+3.97 = 11.86$ \$/hour, 
while in the sub-coalition case it amounts to $8.08$ \$/hour.
This means that $\text{CP}_1$ and $\text{CP}_2$ have an incentive to leave the
grand coalition (in case it has been formed) and to form a
federation including only them or, said in game-theoretic
words, the grand coalition is \emph{unstable} (we discuss this concept in \dcsSecRef{GAMETHEORY1}).
 
This example shows that a more sophisticated solution than simply forming
the grand coalition is necessary in order to ensure stability and that, as 
a general rule, the resulting federation set may include several
federations.

\section{The Cooperative CP Game}\label{GAMETHEORY1}

As illustrated in the previous section, a CP must consider various factors
before deciding whether to join or not a federation.
Among them, the most important ones are:
\begin{itemize}
\item \emph{Stability}: a federation is \emph{stable} if none of its participants
finds that it is more profitable to leave it (e.g., to stay alone or to join
another federation) rather than cooperating with the other ones.
\item \emph{Fairness}: when joining a federation, a CP expects that the resulting
profits are fairly divided among participants.
As unfair division leads to instability, it is necessary to design an allocation
method that ensures fairness.
\end{itemize}

In this section, we model the problem of coalition formation
as a \emph{coalition formation cooperative game with
transferable utility} \cite{Peleg-2007-CooperativeGames,Book_RAY2007}, where each CP
cooperates with the other ones in order to maximize its net profit rate,
and we present a distributed algorithm for solving this problem.

\subsection{Characterization}\label{sec:CP-game}

Given a set $\calg{N}=\{1,2,\ldots,n\}$ of CPs (henceforth also referred to as
the \emph{players}), a \emph{coalition} $\calg{S} \subseteq \calg{N}$ represents
an agreement among the CPs in $\calg{S}$ to act as a single entity (i.e.,
$\calg{S}$ is a federation of CPs).
Specifically, in this paper, a coalition $\calg{S}$ implies that the CPs
belonging to $\calg{S}$ perform a global allocation of their joint workload
$\calg{J}_\calg{S}$ by using as host set the union $\calg{H}_\calg{S}$ of their
host sets.
In other words, the CPs of the coalition act as a single CP with a load that is
the composition of the loads and with a host set that is a composition of the
host sets of the coalition.

A coalition $\calg{S}$ is associated with a \emph{revenue rate}
$r\bigl(\calg{J}_{\calg{S}}\bigr)$ and with an \emph{energy cost rate}
$e\bigl(\calg{J}_{\calg{S}}, \calg{H}_{\calg{S}}\bigr)$.
The revenue rate of $\calg{S}$ is simply the sum of revenue rates of
individual VMs $j \in \calg{J}_{\calg{S}}$, while
$e\bigl(\calg{J}_{\calg{S}}, \calg{H}_{\calg{S}}\bigr)$ can be derived by
minimizing the energy cost resulting from the allocation of the workload
$\calg{J}_{\calg{S}}$ on the host set $\calg{H}_{\calg{S}}$ (we discuss this
in \dcsSecRef{OPTIMIZATION}).

We model the system under study by using the most common form of cooperative
games, i.e., the \emph{characteristic form} \cite{Peleg-2007-CooperativeGames},
where the value of a coalition $\calg{S}$ depends solely on the members of that
coalition, with no dependence on how the players in $\calg{N}\setminus\calg{S}$
are structured (where $\calg{N}\setminus\calg{S}$ denotes the set difference).

Given the above definitions, for the system under study we define the
\emph{coalition value} $v(\cdot)$ as:
\begin{equation} \label{eq:VALUE}
v\bigl(\calg{S}\bigr) = \sum_{i\in\calg{S}}\sum_{j\in\calg{J}_i} r_i\bigl(q(j)\bigr) - e\bigl(\calg{J}_{\calg{S}}, \calg{H}_{\calg{S}}\bigr)
\end{equation}
where $q(j)$ is the class of VM $j$ (see \dcsSecRef{SYSTEM}).

The value $v\bigl(\calg{S}\bigr)$ must be divided, according to a given rule,
among the participants to $\calg{S}$.
The share $x_i(\calg{S})$ of $v\bigl(\calg{S}\bigr)$ received by CP $i$
is its \emph{payoff}, and the vector $\bfg{x}(\calg{S}) \in I \!\! R^{\lvert\calg{N}\rvert}$,
with each component $x_i(\calg{S})$ being the payoff of CP $i \in \calg{S}$, is
the \emph{payoff allocation}.

In our cooperative game, CPs seek to form coalitions in order to increase their
payoffs.
As discussed at the beginning of this section, payoffs should be fairly allocated
so that stable coalitions can form.

In order to ensure fairness in the division of payoffs, we use the
\emph{Shapley value}~\cite{SHAPLEY_53}, a solution
method that is based on the concept of \emph{marginal contribution}.~\footnote{More specifically, we use the \emph{Aumann-Dr\'eze} value~\cite{Aumann-1974-Cooperative}, which is an extension of the Shapley value for games with coalition structures.}

In particular, the Shapley value of player $i$ can be defined as:
\begin{equation}
\label{SHAPLEY}
\phi_i\bigl(v\bigr) = \sum_{\calg{S} \subseteq \calg{N} \setminus \{i\}}
\frac{\bigl\lvert\calg{S}\bigr\rvert!\bigl(n-\bigl\lvert\calg{S}\rvert-1\bigr)!}{n!} \Bigl( v\bigl(\calg{S} \cup \{i\}\bigr)-v\bigl(\calg{S}\bigr)\Bigr)
\end{equation}
where the sum is over all subsets $\calg{S}$ not containing $i$.

It is important to note that since the Shapley value for a player is based on the concept of the player's \emph{marginal contribution} to a coalition (i.e., the change in the worth of a coalition when the player joins to that coalition), the larger is the contribution provided by a player to the coalition, the higher is the payoff allocated to it.
This means that, in a given CP federation, some ``more-contributing'' CP will be rewarded by other ``less-contributing'' CPs to be enrolled in the federation.

Let us now discuss the stability.
In game theory, a typical way to guarantee stability is
to ensure that the allocation of payoffs falls in the so called
\emph{core} \cite{Peleg-2007-CooperativeGames}, that can be intuitively
defined as the set of payoff allocations that guarantees that
no group of players has an incentive to leave the coalition $\calg{N}$ to form
another coalition $\calg{S} \subset \calg{N}$.
It can be shown that a game with a non-empty core contains allocations that can be voluntarily
agreed by all players and are thereby stable, while
in a game with an empty core, some players (or groups of players)
are better off when acting alone than when cooperating all together (the grand coalition $\calg{N}$).

Unfortunately, it can be shown that the core of cooperative CP game defined
above \emph{can be empty} (see \dcsSecRef{sec:proof-emptycore} for a formal
proof of this statement).

To illustrate the effects of a game with an empty core, we return to the
\emph{Scenario $2$} example of \dcsSecRef{EXAM1}, and we compute the coalition values
and payoffs corresponding to all the federations that can be formed  by the
three involved CPs.
From these values,  that are tabulated in \dcsTabRef{tbl:side-pay-ex2}, we note that for the grand coalition
we have that $\phi_1(v)+\phi_2(v)=11.79$, while for the smaller coalition $\{1,2\}$ it results that
$v\bigl(\{1,2\}\bigr)=12.08$.
That is, for $\text{CP}_1$ and $\text{CP}_2$, the coalition $\{1,2\}$ is more
convenient than the grand coalition, or, in other words, the grand coalition
is unstable.
More formally, we can prove that, for the \emph{Scenario $2$}, the core is empty by simply solving the optimization problem shown in \dcsFigRef{opt-core-ex2}, which directly derives from the definition of the core (e.g., see \cite{Peleg-2007-CooperativeGames}), to find one of the imputations (if any) inside the core.
This optimization problem is infeasible (i.e., it does not exist any payoff vector $(x_0,x_1,x_2)$ that satisfies the core conditions) and hence the core is empty.
\begin{figure}
\centering
\begin{small}
\dcsHRule 
\begin{subequations}\label{eq:opt-core-ex2}
\begin{align*}
\text{maximize} & \,\, z = x_0+x_1+x_2 \\
\text{subject to} \nonumber \\
&  x_0 \ge 6.21, \\
&  x_1 \ge 4.15, \\
&  x_2 \ge 4.15, \\
&  x_0 + x_1 \ge 12.08, \\
&  x_0 + x_2 \ge 12.08, \\
&  x_1 + x_2 \ge 8.49, \\
&  x_0 + x_1 + x_2 = 16.27, \\
&  x_0 \ge 0, \\
&  x_1 \ge 0, \\
&  x_2 \ge 0.
\end{align*}
\end{subequations}
\dcsHRule 
\end{small}
\caption{The optimization model to test the emptiness of the core for the \emph{Scenario $2$}}\label{fig:opt-core-ex2}
\end{figure}
These results thus confirm the intuition provided in \dcsSecRef{EXAM1},
i.e., that in some situations the grand coalition may lead to instability.

\begin{table}
\centering
\caption{Scenario $2$: Coalition values and payoffs}\label{tbl:side-pay-ex2}
\begin{tabular}{crc}
\toprule
\multicolumn{1}{c}{Coalition} & \multicolumn{1}{c}{$v(\cdot)$} & \multicolumn{1}{c}{$\phi_i(v)$} \\
\midrule
$\{1\}$     & $ 6.21$ & $\{6.21\}$         \\ 
$\{2\}$     & $ 4.15$ & $\{4.15\}$         \\ 
$\{3\}$     & $ 4.15$ & $\{4.15\}$         \\ 
$\{1,2\}$   & $12.08$ & $\{7.07,5.01\}$    \\ 
$\{1,3\}$   & $12.08$ & $\{7.07,5.01\}$    \\ 
$\{2,3\}$   & $ 8.49$ & $\{4.25,4.25\}$      \\ 
$\{1,2,3\}$ & $16.27$ & $\{7.31,4.48,4.48\}$ \\ 
\bottomrule
\end{tabular}
\end{table}

It is important to note that, in general, the emptiness of the core does not depend on
the particular payoff allocation strategy, but it is instead a peculiarity of the game.
In order to solve this problem, we have thus to resort to a specific type  of
cooperative game, the so called \emph{coalition formation cooperative game} (see
\cite{APT-WITZEL_2009,DrezeGreenberg1980,Book_RAY2007})
that, as shown later, achieves stability by forming independent and disjoint smaller coalitions
when the grand coalition does not form (as in the \emph{Scenario $2$} case discussed above).

More specifically, we consider a class of coalition formation games known
as \emph{hedonic games}~\cite{BOGOMONLAIA-JACKSON_2002,DrezeGreenberg1980}.
Hedonic games can be seen as special cases of cooperative games where a player's
preferences over coalitions depend only on the composition of his coalition.
That is, players prefer being in one coalition rather than in another
one purely based on
who else is in the coalitions they belong.

Let us reformulate our cooperative CP game as a hedonic game.
Given the set $\calg{N}=\{1,2,\ldots,n\}$ of players (i.e., CPs), we define
a \emph{coalition partition} as the set
$\Pi = \{\calg{S}_1, \calg{S}_2, \ldots , \calg{S}_l \}$
that partitions the CPs' set $\calg{N}$.
That is, for $k=1,\ldots, l$, each $\calg{S}_k \subseteq \calg{N}$ is a
disjoint coalition such that $\bigcup_{k=1}^l \calg{S}_k = \calg{N}$ and
$\calg{S}_j\cap\calg{S}_k=\emptyset$ for $j\ne k$.
Given a coalition partition $\Pi$, for any CP $i \in \calg{N}$, we denote by
$\calg{S}_{\Pi}(i)$ the coalition $\calg{S}_k \in \Pi$ such that $i \in \calg{S}_k$.

To set up the coalition formation process, we need to define a
\emph{preference relation} so that each CP can order and compare all the
possible coalitions it belongs and hence it can build preferences over them.
Formally, for any CP $i \in \calg{N}$, a \emph{preference relation}
$\succeq_i$ is defined as a complete, reflexive, and transitive binary relation
over the set of all coalitions that CP $i$ can form (see~\cite{BOGOMONLAIA-JACKSON_2002}).
Specifically, for any CP $i\in\calg{N}$ and given $\calg{S}_1,\calg{S}_2\subseteq\calg{N}$,
the notation $\calg{S}_1 \succeq_i \calg{S}_2$ means that CP $i$
prefers being a member of $\calg{S}_1$ over $\calg{S}_2$ or at least $i$ prefers
both coalitions equally.
The strict counterpart of $\succeq_i$ is denoted by $\succ_i$
and implies that $i$ strictly prefers being a member of $\calg{S}_1$ over $\calg{S}_2$.
Note that the definition of a preference relation is one of the
peculiarities of the coalition formation process.
In general, this relation can be a function
of several parameters, such as the payoffs that the players receive from each
coalition, the approval of the coalition members, and the players' history,
just to name a few.

In our coalition formation CP game, for any CP $i \in \calg{N}$, we use the following preference relation:
\begin{equation}
\label{eq:PREFERENCES}
\calg{S}_1 \succeq_i \calg{S}_2 \Longleftrightarrow f_i(\calg{S}_1) \geq f_i(\calg{S}_2)
\end{equation}
where $\calg{S}_1, \calg{S}_2 \subseteq \calg{N}$ are two coalitions containing CP $i$,
and $f_i(\cdot)$ is a preference function, defined for any CP $i \in \calg{N}$
and any coalition $\calg{S}$ containing $i$, such that:
\begin{equation}
\label{eq:FORMA_f}
f_i(\calg{S}) = \begin{cases}
                  x_i(\calg{S}), & \text{if } \calg{S} \notin h(i), \\
                  -\infty, & \text{otherwise}.
                \end{cases}
\end{equation}
where $x_i(\calg{S})$ is the payoff received by CP $i$ in $\calg{S}$, and $h(i)$
is a history set where CP $i$ stores the identity of the coalition that it
visited and left in the past.
The rationale behind the use of $h(\cdot)$ is to avoid that a CP visit the same
set of coalitions twice (a similar idea has also been used in previously
published work, such as in \cite{Saad-2011-Hedonic,Saad-2011-Coalition}).
Thus, according to \dcsEqRef{FORMA_f}, each CP prefers to join to the coalition
that provides the larger payoff, unless it has already been visited and left in
the past.
The strictly counterpart $\succ_i$ of $\succeq_i$ is defined by replacing $\ge$ with $>$ in \dcsEqRef{PREFERENCES}.

\subsection{A Distributed Coalition Formation Algorithm\label{ALGORITM}}

\begin{figure}
\centering
\begin{small}
\dcsHRule 
\begin{description}
  \item[Step 0: Initialization.] \hfill \\
  At time $t=0$, the CPs are partitioned as:
  \begin{align*}
   \Pi_0 &= \Bigl\{ \bigl\{ 1 \bigr\}, \bigl\{ 2 \bigr\}, \ldots, \bigl\{ n \bigr\} \Bigr\},\\
   h(i) &= \emptyset, \quad \forall i \in \calg{N}.
  \end{align*}
  \item[Step 1: Coalition Formation Stage I.] \hfill \\
  Given the current coalition partition $\Pi_c$, each CP $i$ investigates
  possible hedonic shift operations, in order to look for a coalition
  $\calg{S}_k\in\Pi_c\cup\emptyset$ (if any) such that:
  \begin{equation*}
   \calg{S}_k \cup \{ i \} \succ_i \calg{S}_{\Pi_c}(i).
  \end{equation*}
  \item[Step 2: Coalition Formation Stage II.] \hfill \\
  If such coalition $\calg{S}_k$ is found, CP $i$ decides to perform the hedonic
  shift rule to move to $\calg{S}_k$:
  \begin{enumerate}
  \item CP $i$ updates its history $h(i)$ by adding $\calg{S}_{\Pi_c}(i)$.
  \item CP $i$ leaves its current coalition $\calg{S}_{\Pi_c}(i)$ and joins the
  new coalition $\calg{S}_k$.
  \item $\Pi_c$ is updated:
    \begin{scriptsize}
	\begin{equation*}
     \Pi_{c+1}=\Bigl( \Pi_c \setminus \bigl\{ \calg{S}_{\Pi_c}(i), \calg{S}_k\bigr\} \Bigr) \cup \Bigl\{ \calg{S}_{\Pi_c}(i) \setminus \bigl\{i\bigr\}, \calg{S}_k \cup \bigl\{i\bigr\} \Bigr\}.
	\end{equation*}
    \end{scriptsize}
  \end{enumerate}
  Otherwise, CP $i$ remains in the same coalition so that:
    \begin{scriptsize}
    \begin{equation*}
     \Pi_{c+1}=\Pi_c
    \end{equation*}
    \end{scriptsize}
  \item[Step 3: Coalition Formation Stage III.] \hfill \\
  Repeat Step 1 and Step 2 until all CPs converge to a final partition $\Pi_f$.
\end{description}
\dcsHRule 
\end{small}
\caption{The Distributed Coalition Formation Algorithm}\label{fig:ALGO}
\end{figure}

We are now ready to define our distributed algorithm for coalition formation
that allows each player to decide in a \emph{selfish} way to which coalitions to
join at any point in time.

This algorithm is based on the following \emph{hedonic shift rule}
(see \cite{Saad-2009-Selfish}): given a coalition partition
$\Pi=\{\calg{S}_1, \ldots, \calg{S}_l \}$ on the set $\calg{N}$ and a preference
relation $\succ_i$, any CP $i \in \calg{N}$ decides to leave its current
coalition $\calg{S}_{\Pi}(i)$ 
and join another coalition $\calg{S}_k \in \Pi \cup \emptyset$
(with $\calg{S}_k \neq \calg{S}_{\Pi}(i)$) if and only if
$\calg{S}_k \cup \{i \} \succ_i \calg{S}_{\Pi}(i)$.
The shift rule can be seen as a selfish decision made by a CP to move from its
current coalition to a new one, \emph{regardless of the effects of this move on
the other CPs}.

Whenever a CP $i$ applies this rule, it updates its history set $h(i)$ to store
the coalition $\calg{S}_{\Pi}(i)$ it is leaving.
After the rule is applied, the partition $\Pi$ changes into a new partition
$\Pi'$, such that:
\begin{equation}
\Pi'=\Bigl( \Pi \setminus \bigl\{ \calg{S}_{\Pi}(i), \calg{S}_k\bigr\} \Bigr) \cup \Bigl\{ \calg{S}_{\Pi}(i) \setminus \bigl\{i\bigr\}, \calg{S}_k \cup \bigl\{i\bigr\} \Bigr\}
\end{equation}

Using the hedonic shift rule and the preference relation defined in \dcsEqRef{PREFERENCES} and \dcsEqRef{FORMA_f}, we construct a distributed coalition formation algorithm, shown in
\dcsFigRef{ALGO}.

To implement the proposed algorithm in a real environment, a suitable approach
must be taken.
For a centralized approach, CPs can rely to a central coordinator, to which CPs
communicate their decisions and from which CPs obtain information to update
their local state.
For what regards a distributed approach, suitable techniques for neighbor
discovering, communication and synchronization must be used.
To this end, well-known algorithms exist in distributed and multi-agent systems literature (e.g., see \cite{Coulouris-2011-Distributed,Weiss-2013-MAS}).

It is worth noting that the presented algorithm can be executed at specific instants of time or when new VM requests arrive to CPs, thus making our coalition formation mechanism able to adapt to environmental changes.

We now prove that our algorithm always converges to a stable partition.

\begin{proposition} \label{CONVERGENZA}
Starting from any initial coalition structure $\Pi_0$, the proposed algorithm
always converges to a final partition $\Pi_f$.
\end{proposition}
\begin{proof}
The coalition formation phase can be mapped to a sequence of shift operations.
That is, according to the hedonic shift rule, every shift operation transforms
the current partition $\Pi_c$ into another partition $\Pi_{c+1}$.
Thus, starting from the initial step, the algorithms yields the following
transformations:
\begin{equation} \label{eq:SEQUENZA}
\Pi_0 \rightarrow
\Pi_1 \rightarrow
\cdots \rightarrow
\Pi_c \rightarrow \Pi_{c+1}
\end{equation}
where the symbol $\rightarrow$ denotes the application of a shift operation.
Every application of the shift rule generates two possible cases: (a)
$\calg{S}_k\ne\emptyset$, so it leads to a new coalition partition, or (b)
$\calg{S}_k=\emptyset$, so it yields a previously visited coalition
partition with a non-cooperatively CP (i.e., with a coalition of size $1$).
In the first case, the number of transformations performed by the shift rule is
finite (at most, it is equal to the number of partitions, that is the Bell
number; see~\cite{ROTA_1964}), and hence the sequence in \dcsEqRef{SEQUENZA}
will always terminate and converge to a final partition $\Pi_f$.
In the second case, starting from the previously visited partition, at certain
point in time, the non-cooperative CP must either join a new coalition and yield a
new partition, or decide to remain non-cooperative.
From this, it follows that the number of re-visited partitions will be limited,
and thus, in all the cases the coalition formation stage of the algorithm will
converge to a final partition $\Pi_f$.
\end{proof}

We address the stability of the final partition $\Pi_f$ by using the concept of
\emph{Nash-stability} (see~\cite{BOGOMONLAIA-JACKSON_2002}).
Intuitively, a partition $\Pi$ is considered Nash-stable if no CP has incentive
to move from its current coalition $\calg{S}_{\Pi}(i)$ to join a different
coalition of $\Pi$, or to act alone.
More formally, a partition $\Pi=\{ \calg{S}_1,\ldots,\calg{S}_l \}$ is
\emph{Nash-stable} if $\forall i\in\calg{N}$, $\calg{S}_{\Pi}(i)\succeq_i\calg{S}_k\cup\{i\}$
for all $\calg{S}_k \in \Pi \cup \emptyset$.

Let us show that the partition to which our algorithm converges is Nash-stable.

\begin{proposition} \label{NASH-STABLE}
Any final partition $\Pi_f$ resulting from the algorithm presented in \dcsFigRef{ALGO}
is Nash-stable.
\end{proposition}
\begin{proof}
To show this, we use the proof by contradiction technique.
Assume that the final partition $\Pi_f$ is not Nash-stable.
Consequently, there exists a CP $i \in \calg{N}$ and a coalition
$\calg{S}_k\in\Pi_f\cup\emptyset$ such that $\calg{S}_k\cup\{i\}\succ_i\calg{S}_{\Pi_f}(i)$.
Then, CP $i$ will perform a hedonic shift operation and hence
$\Pi_f \rightarrow \Pi'_f$.
This contradicts the assumption that $\Pi_f$ is the final outcome of our
algorithm.
\end{proof}

The Nash-stability also implies the so called \emph{individual-stability}
(see \cite{BOGOMONLAIA-JACKSON_2002}).
A partition $\Pi=\{ \calg{S}_1,\ldots,\calg{S}_l \}$ is
\emph{individually-stable} if do not exist a player $i \in \calg{N}$ and a
coalition $\calg{S}_k \in \Pi \cup \emptyset$ such that
$\calg{S}_k\cup\{i\}\succ_i\calg{S}_{\Pi}(i)$ and
$\calg{S}_k\cup\{i\}\succeq_j\calg{S}_k$ for all $j \in \calg{S}_k$.

It is worth noting that Nash-stability only captures the notion of stability with respect to movements of single CPs (i.e., no CP has an incentive to unilaterally deviate).
However, it does not guarantee the stability with respect to other aspects that are instead captured by other stability concepts.
For instance, the stability with respect to movements of groups of CPs is captured by the \emph{core}-stability (see \cite{BOGOMONLAIA-JACKSON_2002}), whereby no group of CPs can collectively defect and form a new coalition where each of them is better off.
Unfortunately, the two stability concepts are not related each other, in general (i.e., one stability concept does not necessarily imply the other one).
Moreover, there exist other stronger stability concepts but unfortunately there is not warranty that a satisfying partition does exist (e.g., see \cite{Aziz-2012-Existence}) and the check for the existence is computationally hard (e.g., see \cite{Aziz-2013-Computing}).
Finally, Nash-stability does not guarantee the maximization of the overall net profit (e.g., the \emph{social optimum} in the game-theoretic jargon) \cite{Hasan-2013-Nash}.
Despite all of that, Nash-stability is generally considered a reasonable trade-off.

Thus, we can conclude that our algorithm always converges to a partition $\Pi_f$
which is both Nash-stable and individually stable.

\subsection{Computation of the Coalition Value}\label{OPTIMIZATION}

To use the game-theoretic model discussed in the previous section, we need a way
to find (for a given coalition) the optimal workload allocation (i.e., the
allocation that minimizes the energy cost), that allows us to compute the
coalition value.

To this end, we define a \emph{Mixed Integer Linear Program} (MILP)
modeling the problem of allocating a set $\calg{J}_{\calg{S}}$ of VMs
onto a set $\calg{H}_{\calg{S}}$ of hosts
so that the hourly energy cost is minimized.

We base our MILP on the model described in \cite{BORGETTO2011}, that has been
first revised to improve its computational performance and then extended in
order to incorporate the heterogeneity of physical resources and the energy
cost.

The resulting optimization model is shown in \dcsFigRef{opt-boundedyield}, where
we use the same notation introduced in \dcsSecRef{SYSTEM},\footnote{To ease
readability, we simplify it by denoting with $\calg{J}$ and $\calg{H}$ the cloud
workload and the host set, respectively (i.e., we omit the dependence by
$\calg{S}$).} and we denote with $o(i)$ the function that is $1$ if host
$i$ is powered on and $0$ otherwise (i.e., a function $o:\calg{H} \to \{0,1\}$),
with $L_k$ and $S_k$ the power (in W) consumed during the switch-on and
switch-off operations of a host of class $k$, respectively, with $G_{i_1,i_2,v}$
the hourly cost (in \$/hour) to migrate a VM of class $v$ from CP $i_1$ to CP
$i_2$, with $E_i$ the hourly cost (in \$/Wh) of the energy consumed by a host
belonging to CP $i$, with $c(i)$ the function that gives the CP that owns
host $i$ (i.e, a function $c:\calg{H} \to \calg{N}$), and with $h(j)$ the
function that gives the host where VM $j$ is allocated (i.e., a function
$h:\calg{J} \to \calg{H}$).

In the optimization model we define, for any VM $j \in \calg{J}$ and host
$i \in \calg{H}$, the following decision variables: $b_{ji}$ is a binary
variable that is equal to $1$ if VM $j$ is allocated to host $i$;
$\alpha_{i}$ is a real variable representing the overall fraction of CPU
assigned to all VMs allocated on host $i$; $p_i$ is a binary variable that
is equal to $1$ if host $i$ is powered on.
The objective function $e\bigl(\calg{J},\calg{H}\bigr)$ (hereafter, $e$
for short) represents, for a specific assignment of decision variables, the
hourly energy cost (in \$/hour) due to the power consumption induced by the
federation of CPs to host the given VMs.
\begin{figure}
\centering
\begin{small}
\dcsHRule \\
\begin{subequations}\label{eq:opt-boundedyield}
\begin{align}
\text{minimize} & \,\, e = \sum_{i \in \calg{H}}\Bigl[p_i C_{g(i)}^\text{min} + \alpha_{i}\bigl(C_{g(i)}^\text{max}-C_{g(i)}^\text{min}\bigr) \notag \\
                &  \quad                        + p_i (1-o(i)) L_{g(i)} + (1-p_i) o(i) S_{g(i)}\Bigr]E_{c(i)} \notag \\
				&  \quad                        + \sum_{j\in\calg{J}}b_{ji} G_{c(h(j)),c(i),q(j)} \label{eq:opt-boundedyield-obj}\\
\text{subject to} \nonumber \\
& \sum_{i \in \calg{H}}{b_{ji}} = 1, \qquad \qquad \quad j \in \calg{J}, \label{eq:opt-boundedyield-c4}\\
& \sum_{j\in\calg{J}} b_{ji} \le |\calg{J}| p_i, \,\,  \qquad \quad i \in \calg{H}, \label{eq:opt-boundedyield-c6}\\
& \alpha_{i} \le p_i, \, \qquad \qquad \qquad \quad i \in \calg{H}, \label{eq:opt-boundedyield-c8}\\
& \sum_{j \in \calg{J}}{b_{ji}M_{q(j)g(i)}} \le p_i, \quad i \in \calg{H},\label{eq:opt-boundedyield-c9}\\
& \sum_{j \in \calg{J}}{b_{ji}A_{q(j)g(i)}} = \alpha_{i}, \quad i \in \calg{H}, \label{eq:opt-boundedyield-c10}\\
& b_{ji} \in \lbrace 0,1 \rbrace, \qquad \qquad \qquad j \in \calg{J}, i \in \calg{H}, \label{eq:opt-boundedyield-c1}\\
& \alpha_{i} \in \bigl[0,1\bigr], \qquad \qquad \qquad i \in \calg{H}, \label{eq:opt-boundedyield-c2}\\
& p_{i} \in \lbrace 0,1 \rbrace, \qquad \qquad \qquad \; i \in \calg{H}. \label{eq:opt-boundedyield-c3}
\end{align}
\end{subequations}
\dcsHRule \\
\end{small}
\caption{The VM allocation optimization model}\label{fig:opt-boundedyield}
\end{figure}

The resulting VM allocation is bound to the following constraints:
\begin{itemize}
\item \dcsEqRef{opt-boundedyield-c4} imposes that each VM is hosted by exactly
one host;
\item \dcsEqRef{opt-boundedyield-c6} states that only hosts that are switched on
can have VMs allocated to them; the purpose of these constraints is to avoid
that a VM is allocated to a host that will be powered off;
\item \dcsEqRef{opt-boundedyield-c8} ensures that (1) the CPU resource of a
powered-on host is not exceeded, and (2) that no CPU resource is consumed on a
host that will be powered off;
\item \dcsEqRef{opt-boundedyield-c9} assures that (1) the RAM resource of a
powered-on host is not exceeded, (2) that VMs hosted on that host receive their
required amount of RAM, and (3) that no RAM resource is consumed on a host that
will be powered off;
\item \dcsEqRef{opt-boundedyield-c10} states that all VMs must exactly obtain
the amount of CPU resource they require;
\item \dcsEqRef{opt-boundedyield-c1}, \dcsEqRef{opt-boundedyield-c2}, and
\dcsEqRef{opt-boundedyield-c3} define the domain of decision variables
$b_{ji}$, $\alpha_i$, and $p_i$, respectively.
\end{itemize}
As in \cite{BORGETTO2011}, in order to keep into considerations QoS requirements
related to each class of VMs, we assumed that each VM $j$ exactly obtains the
amount of CPU $\mathit{CPU}_{q(j)}$ and RAM $\mathit{RAM}_{q(j)}$ as defined by
its class $q(j)$ (see \dcsSecRef{SYSTEM}).

\section{Experimental Evaluation}\label{NUMERICAL}

To illustrate the effectiveness of our algorithm for coalition formation,
we perform a set of experiments in which we compute the federation set
for various scenarios including a population of distinct CPs.

In all scenarios,
we consider $4$ CPs, whose infrastructures are characterized as reported in
\dcsTabRef{tbl:exp-cp-conf},
and we use the same host classes, VM classes and VM shares as the ones defined in
\dcsTabRef{tbl:exam1}.
\begin{table}[hbtp]
\centering
\caption{Experimental evaluation -- Configuration of CPs}\label{tbl:exp-cp-conf}
\begin{tabular}{crrr}
\toprule
\multicolumn{1}{c}{CP} & \multicolumn{3}{c}{\# Hosts} \\
                       & \multicolumn{1}{c}{Class-$1$} & \multicolumn{1}{c}{Class-$2$} & \multicolumn{1}{c}{Class-$3$} \\
\midrule
$\text{CP}_1$ & $40$ & $ 0$ & $ 0$ \\ 
$\text{CP}_2$ & $ 0$ & $40$ & $ 0$ \\ 
$\text{CP}_3$ & $ 0$ & $ 0$ & $40$ \\ 
$\text{CP}_4$ & $15$ & $15$ & $10$ \\ 
\bottomrule
\end{tabular}
\end{table}
We also assume that all CPs use the same revenue rate policy, that is they earn
$0.08$ \$/hour for class-$1$ VMs, $0.16$ \$/hour for class-$2$ VMs, and $0.32$
\$/hour for class-$3$ VMs.
Furthermore, without loss of generality, we also assume that the electricity price
is the same for all CPs and it is equal to $0.4$ \$/kWh.

Starting from this configuration, we set up $400$ scenarios that differ from
each other in the workload of the various CPs, in the power state of each host
and in the VM migration costs.
Specifically, in each scenario the workload of each CP is set 
by randomly generating the number of VMs of each class as 
an integer number uniformly distributed in the
$[0,20]$ interval.
In addition, to provide values to function $o(\cdot)$, we randomly generate the
power state (i.e., ON or OFF) of each host according to a Bernoulli distribution
with parameter $0.5$.
Furthermore, the values for $L_k$ and $S_k$, for each host class $k$, are
computed as the product of the electricity price, the maximum power consumption
and the time taken to complete the switch-on or switch-off operation.
This switch-on/-off time is randomly generated for each host class according
to a Normal distribution with mean of $300$ $\mu\text{sec}$ and standard
deviation (S.D.) of $50$ $\mu\text{sec}$ (e.g., see
\cite{Meisner-2009-PowerNap}).
Finally, the VM migration costs $G_{c_1,c_2,k}$ from CP $c_1$ to CP $c_2$ for
each VM class $k$ are computed as the product between the data transfer cost
rate, the data size to transfer and the time to migrate a VM of class $k$ from
CP $c_1$ to CP $c_2$, and assuming that our algorithm activates every $12$ hours.
The data transfer cost rate is taken from the Amazon EC2 data transfer pricing
\cite{EC2} and set to $0.001$ \$/GB.
Furthermore, we suppose that data are persistently transferred during the
migration time at a fixed data rate of $100$ Mbit/sec for all CPs.
For what concerns the migration time, we assume that it is randomly generated
according to a Normal distribution with mean of $277$ sec and S.D.\ of $182$ sec
for VMs of class $1$, with mean of $554$ sec and S.D of $364$ sec for VMs of
class $2$, and with mean of $1108$ sec and S.D.\ of $728$ sec for VMs of class
$3$ (e.g., see \cite{Akoush-2010-Predicting}).
The migration cost between hosts of the same CP is assumed to be negligible.

For each one the above scenario, we compute the federation set of the
involved CPs by using an ad-hoc simulator written in \texttt{C++} and 
interfaced with CPLEX \cite{CPLEX} to solve the various instances of the
optimization model of \dcsSecRef{OPTIMIZATION}.

In the rest of this section, we first present a summary of the performance
obtained by our algorithm over all scenarios, and then, we illustrate its
behavior by showing its run trace for one of these scenarios.

In \dcsFigRef{exp-gains}, we compare the performance of each scenario in terms
of energy saving and net profit obtained with our algorithm with respect to the
\emph{no-federation} case (i.e., when CPs work in isolation).
Specifically, the figures show, for each scenario, the percentage of the
reduction of energy consumption (see \dcsFigRef{exp-gains-en}) and of the increment
of net profit (see \dcsFigRef{exp-gains-pr}) that all CPs obtain when they
federate according to our algorithm with respect to the case of working
individually.
As can be seen from the figures, our algorithm, with respect to always work
non-cooperatively, allows the CPs to reduce the overall consumed energy from
$11.3\%$ to $33.6\%$ (with an average of $21.6\%$), and to increment the overall
net profit from $5.1\%$ to $20.1\%$ (with an average of $10.5\%$).
\begin{figure}
\centering
\centering
\subfloat[][Reduction of energy consumption (\%)\label{fig:exp-gains-en}]{
\centering
\includegraphics[scale=.50]{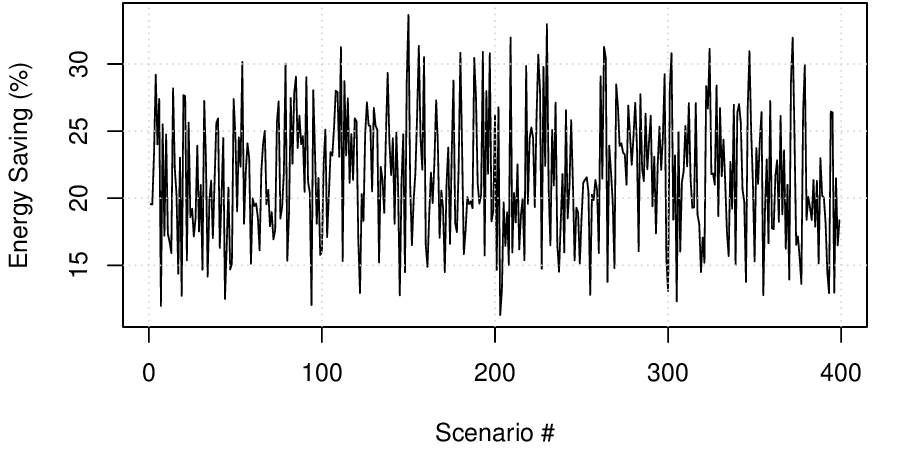}
}
\\
\subfloat[][Increase of net profit (\%)\label{fig:exp-gains-pr}]{
\centering
\includegraphics[scale=.50]{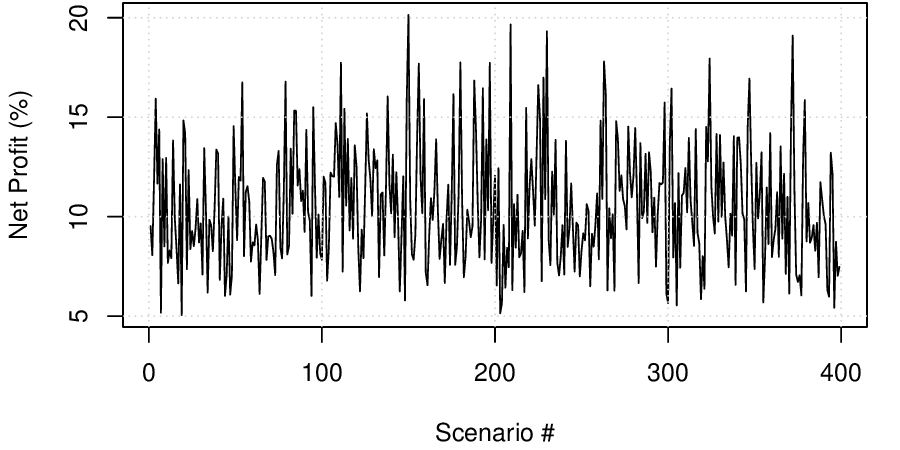}
}
\caption{Performance of our algorithm with respect to the no-federation case}\label{fig:exp-gains}
\end{figure}

We can also analyze the benefits provided by our algorithm from the point of view
of each CP.
Results from our experiments show that, from the CP perspective, the formation
of federations yielded by our algorithm is always non-detrimental.
Specifically, it results that, on
average, the net profit earned by $\text{CP}_1$, $\text{CP}_2$, $\text{CP}_3$
and $\text{CP}_4$ increases by nearly $18.0\%$, $8.5\%$, $22.8\%$ and $4.3\%$ with
respect to the no-federation case, respectively.

Finally, to illustrate how our algorithm works, we present the run trace for
a single scenario, whose characteristics are reported in
\dcsTabRef{tbl:exp-vms}.~\footnote{Due to lack of space, we only report the
number of VMs.
}
We select this scenario to illustrate the behavior of the algorithm when there
are multiple Nash-stable partitions.~\footnote{Note, our algorithm's output is always a single partition.}
\begin{table}
\centering
\caption{Experimental results -- Workload of CPs in the case study}\label{tbl:exp-vms}
\begin{tabular}{crrr}
\toprule
\multicolumn{1}{c}{CP} & \multicolumn{3}{c}{\# VMs} \\
                       & \multicolumn{1}{c}{Class-$1$} & \multicolumn{1}{c}{Class-$2$} & \multicolumn{1}{c}{Class-$3$} \\
\midrule
$\text{CP}_1$ & $ 0$ & $12$ & $13$ \\
$\text{CP}_2$ & $18$ & $ 5$ & $11$ \\
$\text{CP}_3$ & $17$ & $18$ & $11$ \\
$\text{CP}_4$ & $ 3$ & $ 2$ & $ 0$ \\
\bottomrule
\end{tabular}
\end{table}

For this investigation, we show in \dcsTabRef{tbl:exp-coals} all possible
partitions together with the value function $v(\cdot)$ of every coalition inside
each partition, and the corresponding Shapley values.~\footnote{Note, our algorithm does
not necessarily enumerate all of such partitions.}
From the table, we can see that there are two Nash-stable coalitions, namely
$\bigl\{1,2,3,4\bigr\}$ and $\bigl\{ \{1,3\}, \{2,4\} \bigr\}$.
To arrive to one of these partitions, our algorithm works as follows.
Starting from partition $\bigl\{\{1\},\{2\},\{3\},\{4\}\bigr\}$ (i.e.,
every CP works individually), there are two different sequences of hedonic shift
rules:
\begin{itemize}
\item Sequence $\#1$: {\small $\bigl\{\{1\},\{2\},\{3\},\{4\}\bigr\} \stackrel{3}{\longrightarrow} \bigl\{\{1,3\},\{2\},\{4\}\bigr\} \stackrel{2}{\longrightarrow} \bigl\{\{1,2,3\},\{4\}\bigr\} \stackrel{4}{\longrightarrow}\bigl\{\{1,2,3,4\}\bigr\}$}
\item Sequence $\#2$: {\small $\bigl\{\{1\},\{2\},\{3\},\{4\}\bigr\} \stackrel{3}{\longrightarrow} \bigl\{\{1,3\},\{2\},\{4\}\bigr\} \stackrel{2}{\longrightarrow} \bigl\{\{1,3\},\{2,4\}\bigr\}$}
\end{itemize}
where the index on top of each arrow denotes the CP that performs the
corresponding hedonic shift rule.
\begin{table*}
\centering
\caption{Experimental results -- Coalition values and Shapley values for all the 15 different partitions of the case study}\label{tbl:exp-coals}
\begin{tabular}{cccc}
\toprule
\multicolumn{1}{c}{$\Pi=\bigl\{\calg{S}_1,\ldots,\calg{S}_l\bigr\}$} & \multicolumn{1}{c}{$\Bigl\{v\bigl(\calg{S}_1\bigr),\ldots,v\bigl(\calg{S}_l\bigr)\Bigr\}$} & \multicolumn{1}{c}{$\sum_{\calg{S}_i\in\Pi}v(\calg{S}_i)$} & \multicolumn{1}{c}{$\bigr\{\phi_{\calg{S}_1},\ldots,\phi_{\calg{S}_l}\bigr\}$} \\
\midrule
$\Bigl\{\bigl\{1\bigr\},\bigl\{2\bigr\},\bigl\{3\bigr\},\bigl\{4\bigr\}\Bigr\}$ & $\bigl\{4.28,3.45,3.84,0.38\bigr\}$ & $11.95$ & $\Bigl\{\bigl\{4.28\bigr\},\bigl\{3.45\bigr\},\bigl\{3.84\bigr\},\bigl\{0.38\bigr\}\Bigr\}$ \\
$\Bigl\{\bigl\{1,2\bigr\},\bigl\{3\bigr\},\bigl\{4\bigr\}\Bigr\}$               & $\bigl\{8.22,3.84,0.38\bigr\}$      & $12.44$ & $\Bigl\{\bigl\{4.52,3.70\bigr\},\bigl\{3.84\bigr\},\bigl\{0.38\bigr\}\Bigr\}$ \\
$\Bigl\{\bigl\{1,3\bigr\},\bigl\{2\bigr\},\bigl\{4\bigr\}\Bigr\}$               & $\bigl\{9.59,3.45,0.38\bigr\}$      & $13.42$ & $\Bigl\{\bigl\{5.01,4.57\bigr\},\bigl\{3.45\bigr\},\bigl\{0.38\bigr\}\Bigr\}$ \\
$\Bigl\{\bigl\{1\bigr\},\bigl\{2,3\bigr\},\bigl\{4\bigr\}\Bigr\}$               & $\bigl\{4.28,7.82,0.38\bigr\}$      & $12.48$ & $\Bigl\{\bigl\{4.28\bigr\},\bigl\{3.72,4.10\bigr\},\bigl\{0.38\bigr\}\Bigr\}$ \\
$\Bigl\{\bigl\{1,4\bigl\},\bigl\{2\bigr\},\bigl\{3\bigr\}\Bigr\}$               & $\bigl\{4.69,3.45,3.84\bigr\}$      & $11.98$ & $\Bigl\{\bigl\{4.29,0.40\bigr\},\bigl\{3.45\bigr\},\bigl\{3.84\bigr\}\Bigr\}$ \\
$\Bigl\{\bigl\{1\bigr\},\bigl\{2,4\bigr\},\bigl\{3\bigr\}\Bigr\}$               & $\bigl\{4.28,4.33,3.84\bigr\}$      & $12.45$ & $\Bigl\{\bigl\{4.28\bigr\},\bigl\{3.70,0.63\bigr\},\bigl\{3.84\bigr\}\Bigr\}$ \\
$\Bigl\{\bigl\{1\bigr\},\bigl\{2\bigr\},\bigl\{3,4\bigr\}\Bigr\}$               & $\bigl\{4.28,3.45,5.43\bigr\}$      & $13.16$ & $\Bigl\{\bigl\{4.28\bigr\},\bigl\{3.45\bigr\},\bigl\{4.44,0.99\bigr\}\Bigr\}$ \\
$\Bigl\{\bigl\{1,2,3\bigr\},\bigl\{4\bigr\}\Bigr\}$                             & $\bigl\{13.27,0.38\bigr\}$          & $13.65$ & $\Bigl\{\bigl\{5.00,3.70,4.57\bigr\},\bigl\{0.38\bigr\}\Bigr\}$ \\
$\Bigl\{\bigl\{1,2,4\bigr\},\bigl\{3\bigr\}\Bigr\}$                             & $\bigl\{8.69,3.84\bigr\}$           & $12.53$ & $\Bigl\{\bigl\{4.39,3.80,0.50\bigr\},\bigl\{3.84\bigr\}\Bigr\}$ \\
$\Bigl\{\bigl\{1,2\bigr\},\bigl\{3,4\bigr\}\Bigr\}$                             & $\bigl\{8.22,5.43\bigr\}$           & $13.65$ & $\Bigl\{\bigl\{4.52,3.70\bigr\},\bigl\{4.44,0.99\bigr\}\Bigr\}$ \\
$\Bigl\{\bigl\{1,3,4\bigr\},\bigl\{2\bigr\}\Bigr\}$                             & $\bigl\{10.01,3.45\bigr\}$          & $13.46$ & $\Bigl\{\bigl\{4.63,4.78,0.60\bigr\},\bigl\{3.45\bigr\}\Bigr\}$ \\
$\Bigl\{\bigl\{1,3\bigr\},\bigl\{2,4\bigr\}\Bigr\}$                             & $\bigl\{9.59,4.33\bigr\}$           & $13.92$ & $\Bigl\{\bigl\{5.01,4.57\bigr\},\bigl\{3.70,0.63\bigr\}\Bigr\}$ \\
$\Bigl\{\bigl\{1,4\bigr\},\bigl\{2,3\bigr\}\Bigr\}$                             & $\bigl\{4.69,7.82\bigr\}$           & $12.51$ & $\Bigl\{\bigl\{4.29,0.40\bigr\},\bigl\{3.72,4.10\bigr\}\Bigr\}$ \\
$\Bigl\{\bigl\{1\bigr\},\bigl\{2,3,4\bigr\}\Bigr\}$                             & $\bigl\{4.28,8.88\bigr\}$           & $13.16$ & $\Bigl\{\bigl\{4.28\bigr\},\bigl\{3.62,4.36,0.90\bigr\}\Bigr\}$ \\
$\Bigl\{1,2,3,4\Bigr\}$                                                         & $\bigl\{14.01\bigr\}$               & $14.01$ & $\Bigl\{4.78,3.78,4.76,0.68\Bigr\}$ \\
\bottomrule
\end{tabular}
\end{table*}

Regardless what partition is finally selected, from the third column of
\dcsTabRef{tbl:exp-coals} we can also observe that, for this scenario, the
partition value improvement for both Nash-stable partitions with respect to
the non-cooperative behavior (i.e., partition
$\bigl\{\{1\},\{2\},\{3\},\{4\}\bigr\}$) is about $10\%$ for partition
$\bigl\{\{1,3\},\{2,4\}\bigr\}$ and nearly $17\%$ for the grand-coalition.

\section{Related Works} \label{sec:related}

Recently, the concept of cloud federations
\cite{rochwerger:reservoir-computer,moreno:iaas-cloud} has been proposed as a
way to provide individual CPs with more flexibility when allocating on-demand
workloads.
Existing work on cloud federations has been mainly focused on the development
of architectural models for federations \cite{ferrer:optimis}, and of mechanisms
providing specific functionalities (e.g., workload management
\cite{balava,larsson:monitoring}, accounting and billing
\cite{elmroth:accounting}, and pricing
\cite{Hassan-2012-Cooperative,Kunsemoller-2011-Game,mihailescu:pricing,Toosi-2011-Resource}).

To the best of our knowledge, very little work has been carried out to jointly
tackle the problem of dynamically forming stable cloud federations for
energy-aware resource provisioning.
Indeed, much of the existing work only focuses on a single aspect of the problem.
In \cite{Goiri-2012-Economical}, the design and implementation of a VM scheduler
for a federation of CPs is presented.
The scheduler, in addition to manage resources that are local to each CP, is
able to decide when to rent resources from other CPs, when to lease own idle
resources to other CPs, and when to turn on or off local physical resources.
Unlike our work, this work does not consider the problem of forming stable CP
federations.
In \cite{Niyato-2011-Cooperative}, a cooperative game-theoretic model for
federation formation and VM management is proposed.
In this work, the federation formation among CPs is analyzed using the concept
of network games, but the energy minimization problem is not considered.

In \cite{Mashayekhy-2012-Coalitional}, a profit-maximizing game-based mechanism
to enable dynamic cloud federation formation is proposed.
The dynamic federation formation problem is modeled as a hedonic game (like our
approach), and the federations are computed by means of a merge-split algorithm.
There are several important differences between this and our works: (1) we focus
on the stability of individuals rather than of groups, (2) we
propose a decentralized algorithm, (3) we demonstrate the stability of the
obtained federations, and (4) we use the Shapley value instead of the normalized
Banzhaf value (as in \cite{Mashayekhy-2012-Coalitional}), since the latter does
not satisfy some important properties \cite{VanDenBrink-1998-Axiomatizations}.

In \cite{Samaan-2013-Novel}, the problem of sharing unused capacity in a
federation of CPs for VM spot market is formulated as a non-cooperative repeated
game.
Specifically, by using a Markov model to predict future non-spot workload, the
authors introduce a set of capacity sharing strategies that maximize the
federation's long-term revenue and propose a dynamic programming algorithm to
find the allocation rules needed to achieve it.
Our work can complement this approach by providing a solution to the formation
of CP federations for non-spot VM instances.

\section{Conclusions and Future Works\label{CONCL}}

This paper investigates a novel dynamic federation scheme among 
a set of CPs. 
To this end, we propose a cooperative game-theoretic framework to study the
federation formation problem, and a mathematical optimization model to allocate
CP workload in an energy-aware fashion, in order to reduce CP energy costs.

In the proposed scheme, we model the cooperation among the CPs
as a coalition game with transferable utility and we devise a distributed
hedonic shift algorithm for coalition formation.
With the proposed algorithm, each CP individually decides whether to leave
the current coalition to join a different one according to his preference,
meanwhile improving the perceived net profit.
Furthermore, we prove that the proposed algorithm converges to a Nash-stable
partition which determines the resulting coalition structure.
Numerical results show the effectiveness of our approach.

The future developments of this research is following several directions.
First of all, we would like to enhance the coalition value function in order to
account for possible request losses due to lack of physical resources.
Furthermore, we want to improve the game-theoretic and optimization models in
order to include costs in terms of loss of revenues as well as other aspects
like the ones related to trustworthiness among CPs.


As a second direction, we plan to integrate the long-term
resource provisioning solution proposed in this paper with other short-term
and medium-term resource management strategies (e.g.,
\cite{Albano-2013,Guazzone-2012-Exploiting}) to improve resource utilization and meet
application-level performance requirements, and with techniques for incremental
VM migration (e.g., \cite{Verma-2008-pMapper}).

Finally, we want to implement and validate the proposed algorithm in a real testbed.

\bibliographystyle{IEEEtran}
\bibliography{IEEEabrv,paper-full}

\appendix

\section{The Core of the Cooperative CP Game can be Empty}\label{sec:proof-emptycore}

In this section, we present a more formal proof of the possible emptiness of the core of the cooperative CP game defined in \dcsSecRef{sec:CP-game}.
To do so, we use the proof by construction technique, by providing an instance of the game for which the core is empty.

In cooperative game theory, the \emph{Bondareva-Shapley theorem} provides the necessary and sufficient conditions for the non-emptiness of the core solution concept~\cite{Peleg-2007-CooperativeGames}.
\begin{thm}[Bondareva-Shapley theorem]\label{thm:bondareva-shapley}
Given a cooperative game $\langle\calg{N},v\rangle$, the core of $\langle\calg{N},v\rangle$ is non-empty if and only if for every function $\alpha : 2^\calg{N} \setminus \{\emptyset\} \to [0,1]$ where
\[
\forall i \in \calg{N} : \sum_{\calg{S} \in 2^\calg{N} : \; i \in \calg{S}} \alpha\bigl(\calg{S}\bigr) = 1
\]
the following condition holds:
\begin{equation}\label{eq:bondareva-shapley}
\sum_{\calg{S} \in 2^\calg{N}\setminus\{\emptyset\}} \alpha\bigl(\calg{S}\bigr) v\bigl(\calg{S}\bigr) \leq v\bigl(\calg{N}\bigr). 
\end{equation}
\end{thm}

We now show that for the cooperative CP game there exists at least one counterexample that violates the conditions \dcsEqRef{bondareva-shapley} of the Bondareva-Shapley theorem.

Let us consider a simple scenario consisting of the same host and VM classes as defined in \dcsSecRef{EXAM1}, and of three CPs, whose characteristics are reported in \dcsTabRef{tbl:app-cipcat}.
\begin{table}
\centering
\caption{Configuration of CPs}\label{tbl:app-cipcat}
\begin{tabular}{crrr}
\toprule
CP & \multicolumn{3}{c}{\# Hosts} \\
   & \multicolumn{1}{c}{Class-$1$} & \multicolumn{1}{c}{Class-$2$} & \multicolumn{1}{c}{Class-$3$} \\
\midrule
$\text{CP}_1$ & $0$ & $2$ & $0$ \\
$\text{CP}_2$ & $1$ & $0$ & $0$ \\
$\text{CP}_3$ & $1$ & $0$ & $0$ \\
\midrule
CP & \multicolumn{3}{c}{\# VMs} \\
   & \multicolumn{1}{c}{Class-$1$} & \multicolumn{1}{c}{Class-$2$} & \multicolumn{1}{c}{Class-$3$} \\
\midrule
$\text{CP}_1$ & $0$ & $4$ & $0$ \\
$\text{CP}_2$ & $0$ & $1$ & $0$ \\
$\text{CP}_3$ & $0$ & $1$ & $0$ \\
\midrule
CP & \multicolumn{3}{c}{Energy Cost}\\
   & \multicolumn{3}{c}{(\$/kWh)} \\
\midrule
$\text{CP}_1$ & \multicolumn{3}{r}{$0.4$} \\
$\text{CP}_2$ & \multicolumn{3}{r}{$0.4$} \\
$\text{CP}_3$ & \multicolumn{3}{r}{$0.4$} \\
\bottomrule
\end{tabular}
\end{table}

We build the cooperative CP game $\langle\calg{N},v\rangle$ where $\calg{N}=\{1,2,3\}$ is the set of CPs and $v(\cdot)$ is the same characteristic function defined in \dcsEqRef{VALUE}.
In \dcsTabRef{tbl:app-coalitions}, we show the enumeration of all possible CP coalitions for this game along with their values.
To compute the coalition value we use the same revenue rate described in \dcsSecRef{NUMERICAL}, that is $0.08$ \$/hour for class-$1$ VMs, $0.16$ \$/hour for
class-$2$ VMs, and $0.32$ \$/hour for class-$3$ VMs.
\begin{table}
\centering
\caption{Values of $v(\cdot)$ for CPs coalitions}\label{tbl:app-coalitions}
\begin{tabular}{cr}
\toprule
Coalitions $\calg{S}$ & \multicolumn{1}{c}{$v\bigl(\calg{S}\bigr)$} \\
                      & \multicolumn{1}{c}{(\$/hour)}\\
\midrule
$\bigl\{1\bigr\}$ & $0.345$ \\
$\bigl\{2\bigr\}$ & $0.095$ \\
$\bigl\{3\bigr\}$ & $0.095$ \\
$\bigl\{1,2\bigr\}$ & $0.513$ \\
$\bigl\{1,3\bigr\}$ & $0.513$ \\
$\bigl\{2,3\bigr\}$ & $0.225$ \\
$\bigl\{1,2,3\bigr\}$ & $0.623$ \\
\bottomrule
\end{tabular}
\end{table}

Let us choose as function $\alpha(\cdot)$ in \ref{thm:bondareva-shapley} the following function:
\begin{equation*}
\alpha\bigl(\calg{S}\bigr) = \begin{cases}
								\frac{1}{2}, & \calg{S}\in\biggl\{\bigl\{1,2\bigr\},\bigl\{1,3\bigr\},\bigl\{2,3\bigr\}\biggr\},\\
								0, & \text{otherwise}.
							 \end{cases}
\end{equation*}
If the core is non-empty, \dcsEqRef{bondareva-shapley} would hold.
However, it results that:
\begin{equation*}
v\bigl(\{1,2,3\}\bigr)<\frac{1}{2}\cdot\biggl(v\bigl(\{1,2\}\bigr)+v\bigl(\{1,3\}\bigr)+v\bigl(\{2,3\}\bigr)\biggr)
\end{equation*}
That is:
\begin{align*}
0.623 &< \frac{1}{2}\cdot(0.513+0.513+0.225),\\
0.623 &< 0.625.
\end{align*}
which clearly violates the conditions \dcsEqRef{bondareva-shapley} of the Bondareva-Shapley theorem and hence the core for this game is empty.
\qed

\end{document}